\documentclass[11pt]{article}
\usepackage{amsmath,amsthm,amssymb, fullpage,  graphicx}
\usepackage[ruled]{algorithm2e}

\def\comment#1{}

\newtheorem{definition}{\bf Definition}[section]
\newtheorem{example}{\bf Example}[section]
\newtheorem{theorem}{\bf Theorem}[section]
\newtheorem{lemma}{\bf Lemma}[section]

\newcommand{\qual}{q}
\newcommand{\E}{\mathrm{E}}
\newcommand{\Ev} {\E_{\vec{v}}}

\newcommand{\lv} {\underline{v}}

\newcommand{\vi}  {v_{-i}}

\title{Pricing Ad Slots with Consecutive Multi-unit Demand \\[0.1in]}

\author{
Xiaotie Deng\thanks{Department of Computer Science, Shanghai Jiaotong University, China.
Email: {\tt dengxiaotie@gmail.com}
Supported by the National Science Foundation of China (Grant No. 61173011) and a Project 985 grant of Shanghai Jiaotong University
} \\
\and
Paul W. Goldberg\thanks{Department of Computer Science, Oxford University, UK. Email: {\tt Paul.Goldberg@cs.ox.ac.uk} Supported by EPSRC project EP/K01000X/1 ``Efficient Algorithms for Mechanism Design without Monetary Transfer"}\\
\and
Yang Sun\thanks{Department of Computer Science, City University of Hong Kong, Hong Kong. Email: {\tt sunyang@live.hk}}\\
\and
Bo Tang\thanks{Department of Computer Science, University of Liverpool, UK. Email: {\tt Bo.Tang@liverpool.ac.uk}}\\
\and
Jinshan Zhang\thanks{Department of Computer Science, University of Liverpool, UK. Email: {\tt Jinshan.Zhang@liverpool.ac.uk} Supported by EPSRC project EP/K01000X/1 ``Efficient Algorithms for Mechanism Design without Monetary Transfer"}\\
}

\begin{document}

\maketitle
\thispagestyle{empty}

\begin{abstract}
We consider the optimal pricing problem for a model of the rich
media advertisement market, as well as other related applications.
In this market, there are multiple buyers (advertisers), and items (slots) that are
arranged in a line such as a banner on a website.  Each buyer
desires a particular number of {\em consecutive} slots and has a
per-unit-quality value $v_i$ (dependent on the ad only) while each
slot $j$ has a quality $q_j$ (dependent on the position only such as
click-through rate in position auctions).  Hence, the valuation of
the buyer $i$ for item $j$ is $v_iq_j$. We want to decide the
allocations and the prices in order to maximize the total revenue of
the market maker.

A key difference from the traditional position auction is the
advertiser's requirement of a fixed number of consecutive slots.
Consecutive slots may be needed for a large size rich media ad.
We study three major pricing mechanisms, the Bayesian
pricing model, the maximum revenue market equilibrium model and an envy-free solution model.
Under the Bayesian model, we design a polynomial time computable
truthful mechanism which is optimum in revenue.  For the market
equilibrium paradigm, we find a polynomial time algorithm to obtain
the maximum revenue market equilibrium solution. In envy-free settings,
an optimal solution is presented when the buyers have the same demand for the number of consecutive slots.
We conduct a simulation that compares the revenues from the above schemes
and gives convincing results.
\end{abstract}

\newpage

\section{Introduction}

Ever since the pioneering studies on pricing protocols for sponsored search advertisement,
especially with the generalized second price auction (GSP), by Edelman, Ostrovsky, and
Schwarz~\cite{EOS07}, as well as Varian~\cite{Va07},
market making mechanisms have attracted
much attention from the research community in understanding
their effectiveness for the revenue maximization task facing
 platforms providing Internet advertisement services.
In the traditional advertisement setting, advertisers negotiate ad
presentations and prices with website publishers directly.
An automated pricing mechanism simplifies this process by creating a bidding game for the
buyers of advertisement space over an IT platform. It creates a complete competition environment for
the price discovery process. Accompanying the explosion of the online
advertisement business, there is a need to have a complete picture on
what pricing methods to use in practical terms for both advertisers and Ad space providers.

In addition to search advertisements, display advertisements have been widely used
in webpage advertisements. They have a rich format of displays such as text ads and rich media ads.
Unlike sponsored search, there is a lack of systematic studies on its working mechanisms
for decision makings.
The market maker faces a combinatorial problem of whether to assign a large space to one
large rich media ad or multiple small text ads, as well as how to decide on the prices charged
to them.
We present a study of the allocation and pricing mechanisms for displaying slots in
this environment where some buyers would like to have one slot and others may want several
consecutive slots in a display panel.
In addition to webpage ads, another motivation of our study is TV advertising where
inventories of a commercial break are usually divided into slots of
a few seconds each, and slots have various qualities measuring their
expected number of viewers and the corresponding attractiveness.

We discuss three types of mechanisms and consider the revenue maximization problem under
these mechanisms, and compare their effectiveness in revenue maximization under a dynamic
setting where buyers may change their bids to improve their utilities.
Our results make an important step toward the understanding of the advantages and
disadvantages of their uses in practice. Assume the ad supplier divides the ad
space into small enough slots (pieces) such that each advertiser is
interested in a position with a fixed number of $consecutive$ pieces.
In modelling values to the advertisers, we modify the position auction model
from the sponsored search market~\cite{EOS07,Va07} where each ad slot is measured by the Click
Through Rates (CTR), with users' interest expressed by a click on an ad.
Since display advertising is usually sold on a per impression
(CPM) basis instead of a per click basis (CTR), the quality factor of
an ad slot stands for the expected impression it will brings in unit of time.
Unlike in the traditional position auctions, people may have
varying demands (need different spaces to display their ads) in a rich
media ad auction for the market maker to decide on slot allocations and their prices.

We will lay out the the specific system parameters and present
our results in the following subsections.

\subsection{Our Modeling Approach}

We have a set of {\em buyers} (advertisers) and a set of {\em items} to be sold (the
ad slots on a web page). We address the challenge of computing prices that satisfy
certain desirable properties. Next we describe the elements of the model in more detail.

\begin{itemize}
\item[$\bullet$]{\bf Items.} Our model considers the geometric
  organization of ad slots, which commonly has the slots arranged
  in some sequence (typically, from top to bottom in the right-hand
  side of a web page). The slots are of variable quality. In the study
  of sponsored search auctions, a standard assumption is that the
  quality (corresponding to click-through rate) is highest at the
  beginning of the sequence and then monotonically decreases. Here we
  consider a generalization where the quality may go down and up,
  subject to a limit on the total number of local maxima (which we
  call {\em peaks}), corresponding to focal points on the web page. As
  we will show later, without this limit the revenue maximization
  problem is NP-hard.
\item[$\bullet$]{\bf Buyers.} A buyer (advertiser) may want to
  purchase multiple slots, so as to display a larger ad. Note that
  such slots should be {\em consecutive} in the sequence.  Thus, each
  buyer $i$ has a fixed {\em demand} $d_i$, which is the number of
  slots she needs for her ad. Two important aspects of this are
\begin{itemize}
\item[$\diamond$]{\em sharp} multi-unit demand, referring to the fact
  that buyer $i$ should be allocated $d_i$ items, or none at all;
  there is no point in allocating any fewer
\item[$\diamond$]{\em consecutiveness} of the allocated items, in the
  pre-existing sequence of items.
\end{itemize}
These constraints give rise to a new and interesting combinatorial
pricing problem.
\item[$\bullet$]{\bf Valuations.} We assume that each buyer $i$ has a
  parameter $v_i$ representing the value she assigns to a slot of unit
  quality. Valuations for multiple slots are additive, so that a buyer
  with demand $d_i$ would value a block of $d_i$ slots to be their
  total quality, multiplied by $v_i$.  This valuation model has been
  considered by Edelman et al.~\cite{EOS07} and Varian~\cite{Va07} in
  their seminal work for keywords advertising.
\end{itemize}

{\bf Pricing mechanisms.}
Given the valuations and demands from the buyers, the market maker
decides on a price vector for all slots and an allocation of slots to
buyers, as an output of the market. The question is one of which
output the market maker should choose to achieve certain objectives.
We consider two approaches:
\begin{itemize}
\item[$\bullet$]{\bf Truthful mechanism} whereby the buyers report their demands (publicly known) and values (private)
to the market maker; then prices are set in such a way as to ensure that the buyers
have the incentive to report their true valuations. We give a revenue-maximizing
approach (i.e., maximizing the total price paid), within this framework.
\item[$\bullet$]{\bf Competitive equilibrium} whereby we prescribe certain constraints on the
prices so as to guarantee certain well-known notions of fairness and envy-freeness.
\item[$\bullet$]{\bf Envy-free solution} whereby we prescribe certain constraints on the
prices and allocations so as to achieve envy-freeness, as explained below.
\end{itemize}

The mechanisms we exhibit are computationally efficient. We also
performed experiments to compare the revenues obtained from these three
mechanisms.

\subsection{Related Works}

The theoretical study of position auctions (of a single slot) under the
generalized second price auction was initiated in
\cite{EOS07,Va07}. There has been a series of studies of position
auctions in deterministic settings \cite{La06}. Our consideration of
position auctions in the Bayesian setting fits in the general one
dimensional auction design framework. Our study considers continuous
distributions on buyers' values. For discrete distributions,
\cite{CDW12} presents an optimal mechanism for budget constrained
buyers without demand constraints in multi-parameter settings and very
recently they also give a general reduction from revenue to welfare
maximization in \cite{CDW12-2}; for buyers with both budget
constraints and demand constraints, $2$-approximate mechanisms
\cite{Al11} and $4$-approximate mechanisms \cite{BGGM10} exist in the
literature.

There are extensive studies on multi-unit demand in economics, see for
example~\cite{AC96,CP07,EWK98}.  In an earlier paper~\cite{CDGZ12} we
considered sharp multi-unit demand, where a buyer with demand $d$
should be allocated $d$ items or none at all, but with no further
combinatorial constraint, such as the consecutiveness constraint that
we consider here.  The sharp demand setting is in contrast with a
``relaxed'' multi-unit demand (i.e., one can buy a subset of at most
$d$ items), where it is well known that the set of competitive
equilibrium prices is non-empty and forms a distributive
lattice~\cite{GS99,SS71}.  This immediately implies the existence of
an equilibrium with maximum possible prices; hence, revenue is
maximized. Demange, Gale, and Sotomayor~\cite{DGS86} proposed a
combinatorial dynamics which always converges to a revenue maximizing
(or minimizing) equilibrium for unit demand; their algorithm can be
easily generalized to relaxed multi-unit demand. A strongly related work
to our consecutive settings is the work of Rothkopf
et al.~\cite{RPH98}, where the authors presented a dynamic programming
approach to compute the maximum social welfare of consecutive settings
when all the qualities are the same. Hence, our dynamic programming
approach for general qualities in Bayesian settings is a non-trivial
generalization of their settings.

\subsection{Organization}

This paper is organized as follows. In Section~\ref{sec:pre} we
describe the details of our rich media ads model and the related
solution concepts.  In Section~\ref{sec:single}, we study the problem
under the Bayesian model and provide a Bayesian Incentive Compatible
auction with optimal expected revenue for the special case of the
single peak in quality values of advertisement positions. Then in
Section \ref{sec:general}, we extend the optimal auction to the case
with limited peaks/valleys and show that it is NP-hard to maximize
revenue without this limit. Next, in Section \ref{sec:eq}, we turn to
the full information setting and propose an algorithm to compute the
competitive equilibrium with maximum revenue. In Section \ref{sec:envy},
NP-hardness  of envy-freeness for consecutive multi-unit demand buyers is shown.
We also design a polynomial time solution for the special case where all advertisers
demand the same number of ad slots.
The simulation is presented in Section \ref{sec:simulation}.


\section{Preliminaries}
\label{sec:pre}
In our model, a rich media advertisement instance consists of $n$
advertisers and $m$ advertising slots.
Each slot $j\in \{ 1,\ldots, m\}$ is associated with a number $q_j$ which can
be viewed as the quality or the desirability of the slot. Each advertiser
(or buyer) $i$ wants to display her own ad that occupies $d_i$
consecutive slots on the webpage. In addition, each buyer has a
private number $v_i$ representing her valuation and thus, the $i$-th
buyer's value for item $j$ is $v_{ij}=v_iq_j$.

Throughout this thesis, we will often say that slot $j$ is assigned to
a buyer set $B$ to denote that $j$ is assigned to some buyer in
$B$. We will call the set of all slots assigned to $B$ the allocation
to $B$. In addition, a buyer will be called a winner if he succeeds in
displaying his ad and a loser otherwise. We use the standard
notation $[s]$ to denote the set of integers from $1$ to $s$,
i.e. $[s]=\{1,2,\ldots,s\}$. We sometimes use $\sum_i$ instead of
$\sum_{i\in[n]}$ to denote the summation over all buyers and $\sum_j$
instead of $\sum_{j\in[m]}$ for items, and the terms $\E_{\mathbf{v}}$ and
$\E_{\vi}$ are short for $\E_{\mathbf{v}\in \mathbf{V}}$ and $\E_{\vi\in V_{-i}}$.

The vector of all the buyers' values is denoted by $\vec{v}$ or
sometimes $(v_i;\vi)$ where $\vi$ is the joint bids of all bidders
other than $i$. We represent a feasible assignment by a vector
$\vec{x}=(x_{ij})_{i,j}$, where $x_{ij}\in\{0,1\}$ and
$x_{ij}=1$ denotes item $j$ is assigned to buyer
$i$. Thus we have $\sum_ix_{ij}\leq 1$ for every item $j$. Given a
fixed assignment $x$, we use $t_i$ to denote the
quality of items that buyer $i$ is assigned, precisely,
$t_i=\sum_jq_jx_{ij}$. In general, when $x$ is a function of buyers'
bids $\vec{v}$, we define $t_i$ to be a function of $\vec{v}$ such
that $t_i(\vec{v})=\sum_jq_jx_{ij}(\vec{v})$.

When we say that slot qualities have a single peak,
we mean that there exists a peak slot $k$ such that for any slot $j<k$ on the left
side of $k$, $q_j\ge q_{j-1}$ and for any slot $j>k$ on the right side
of $k$, $q_j\ge q_{j+1}$.
\subsection{Bayesian Mechanism Design}
Following the work of \cite{My81},
we assume that all buyers' values are distributed
independently according to publicly known bounded distributions.
The distribution of each buyer $i$ is represented by a Cumulative
Distribution Function (CDF) $F_i$ and a Probability Density Function
(PDF) $f_i$. In addition, we assume that the concave closure or convex
closure or integration of those functions can be computed efficiently.

An auction $M=(\vec{x},\vec{p})$ consists of an allocation function $\vec{x}$
and a payment function $\vec{p}$.
$\vec{x}$ specifies the allocation of items to buyers and $\vec{p} = (p_i)_i$
specifies the buyers' payments, where both $\vec{x}$ and $\vec{p}$ are functions of
the reported valuations $\vec{v}$. Our objective is to maximize the expected revenue of the
mechanism is $Rev(M)=\Ev\left[\sum_ip_i(\vec{v})\right]$ under Bayesian incentive
compatible mechanisms.
\begin{definition}
  A mechanism $M$ is called {\em Bayesian Incentive Compatible} (BIC)
  iff the following inequalities hold for all $i,v_i,v'_i$.
\begin{equation}
  \label{eq:BIC}
  \E_{\vi}[v_it_i(\vec{v})-p_i(\vec{v})] \geq \E_{\vi}\left[v_it_i(v'_i;\vi)-p_i(v'_i;\vi)\right]
\end{equation}
Besides, we say $M$ is {\em Incentive Compatible} if $M$ satisfies a
stronger condition that $v_it_i(\vec{v})-p_i(\vec{v})\geq
v_it_i(v'_i;\vi)-p_i(v'_i;\vi)$, for all $\vec{v}, i, v'_i$,
\end{definition}

To put it in words, in a BIC mechanism, no player can improve her
\emph{expected} utility (expectation taken over other players' bids)
by misreporting her value. An IC mechanism satisfies the stronger requirement
that no matter what the other players declare, no player has incentives to
deviate.

\subsection{Competitive Equilibrium and Envy-free Solution}

In Section \ref{sec:eq}, we study the revenue maximizing competitive
equilibrium and envy-free solution in the full information setting instead of the Bayesian setting.
An outcome  of the market is a pair $(\vec{X},\vec{p})$, where $\vec{X}$ specifies
an allocation of items to buyers and $\vec{p}$ specifies prices paid.
Given an outcome $(\vec{X},\vec{p})$, recall $v_{ij}=v_iq_j$, let
$u_i(\vec{X},\vec{p})$ denote the $utility$ of $i$.

\begin{definition}\label{Def-EF}
  A tuple $(\vec{X},\vec{p})$ is a {\em consecutive envy-free pricing} solution
  if every buyer is consecutive envy-free,  where a buyer $i$ is consecutive
  envy-free if the following conditions are satisfied:
\begin{itemize}
\item[$\bullet$]if $X_i\neq \emptyset$, then (i) $X_i$ is $d_i$ consecutive
  items.  $u_i(\vec{X},\vec{p})=\sum\limits_{j\in
    X_i}(v_{ij}-p_j)\geq 0$, and (ii) for any other subset of
  consecutive items $T$ with $|T|= d_i$,
  $u_i(\vec{X},\vec{p})=\sum\limits_{j\in X_i}(v_{ij}-p_j)\geq
  \sum\limits_{j\in T}(v_{ij}-p_{j})$;
\item[$\bullet$]if $X_i=\emptyset$ (i.e., $i$ wins nothing), then, for any
  subset of consecutive items $T$ with $|T|= d_i$, $\sum\limits_{j\in
    T}(v_{ij}-p_j)\leq 0$.
\end{itemize}
\end{definition}
In the literature, there have been two other types of envy-free concepts, namely,
sharp item envy-free \cite{CDGZ12} and bundle envy-free \cite{FFLS12}.
Sharp item envy-free requires that each buyer would not envy a
bundle of items with the number of her demand while bundle envy-free
illustrates that no one would envy the bundle bought by any other
buyer. From the definition of those three envy-free concepts, we have the
following inclusive relations: \\ $\text{sharp item
    envy-free}$ $\Rightarrow$ $\text{(sharp) bundle envy-free}$,
\\ $\text{consecutive envy-free}
\Rightarrow\text{(consecutive) bundle envy-free}$
\begin{example}[Three types of envy-freeness]\label{ex-three-envy-free}
  Suppose there are two buyers $i_1$ and $i_2$ with per-unit-quality
  $v_{i_1}=10$, $v_{i_2}=8$ and $d_{i_1}=1$, $d_{i_2}=2$. The item
  $j_1$, $j_2$, $j_3$ with quality as $q_{j_1}=q_{j_3}=1$ and
  $q_{j_2}=3$. The optimal solution of the three types of
  envy-freeness are as follows:
\begin{itemize}
\item[$\bullet$] The optimal consecutive envy-free solution, $X_{i_1}=\{j_3\}$, $X_{i_2}=\{j_1,j_2\}$ and $p_{j_1}=p_{j_3}=6$ and $p_{j_2}=26$ with total revenue $38$;
\item[$\bullet$] Optimal sharp item envy-free solution, $X_{i_1}=\{j_2\}$, $X_{i_2}=\{j_1,j_3\}$ and $p_{j_1}=p_{j_3}=8$ and $p_{j_2}=28$ with total revenue $44$;
\item[$\bullet$]Optimal (relaxed) bundle envy-free solution, $X_{i_1}=\{j_2\}$, $X_{i_2}=\{j_1,j_3\}$ and $p_{j_1}=p_{j_3}=8$ and $p_{j_2}=30$ with total revenue $46$;
\end{itemize}
\end{example}

\begin{definition}(Competitive Equilibrium)
  We say an outcome of the market $(\vec{X},\vec{p})$ is a {\em competitive
  equilibrium} if it satisfies two conditions.
  \begin{itemize}
  \item[$\bullet$]$(\vec{X},\vec{p})$ must be consecutive envy-free.
  \item[$\bullet$]The unsold items must be priced at zero.
  \end{itemize}
\end{definition}

We are interested in the revenue maximizing competitive equilibrium and envy-free solutions.

It is well known that a competitive equilibrium always exists for unit
demand buyers (even for general $v_{ij}$ valuations)~\cite{SS71}.
For our consecutive multi-unit demand model, however, a competitive equilibrium
may not always exist as the following example shows.

\begin{example}[Competitive equilibrium may not exist]\label{ex-max-CE-not-exist}
There are two buyers $i_1,i_2$ with values $v_{i_1}=10$ and $v_{i_2}=9$, respectively.
Let their  demands be $d_{i_1}=1$ and $d_{i_2}=2$, respectively.
Let the seller have
two items $j_1,j_2$, both with the unit quality $\qual_{j_1}=\qual_{j_2}=1$.
If $i_1$ wins an item,
without loss of generality, say $j_1$, then $j_2$ is unsold and
$p_{j_2}=0$; by envy-freeness of $i_1$, we have $p_{j_1} = 0$ as well.
Thus,
$i_2$ envies the bundle $\{j_1,j_2\}$.
On the other hand, if $i_2$ wins both items, then
$p_{j_1}+p_{j_2}\le v_{i_2j_1}+v_{i_2j_2}=18$, implying that $p_{j_1}\le
9$ or $p_{j_2}\le 9$.
Therefore, $i_1$ is not envy-free. Hence, there is no
competitive equilibrium in the given instance.
\end{example}

In the unit demand case, it is well-known that the set of equilibrium prices
forms a distributive lattice; hence, there exist extremes which correspond
to the maximum and the minimum equilibrium price vectors. In our consecutive
demand model, however, even if a competitive equilibrium
exists, maximum equilibrium prices may not exist.

\begin{example}[Maximum equilibrium need not exist]\label{example-no-max-eq}
There are two buyers $i_1,i_2$ with values $v_{i_1}=10,v_{i_2}=1$ and
demands $d_{i_1}=2,d_{i_2}=1$, and two items $j_1,j_2$ with unit quality
$\qual_{j_1}=\qual_{j_2}=1$. It can be seen that allocating the two items to
$i_1$ at prices $(19,1)$ or $(1,19)$ are both revenue maximizing
equilibria; but there is no equilibrium price vector which is at least
both $(19,1)$ and $(1,19)$.
\end{example}

Because of the consecutive multi-unit demand,
it is possible that some items are
`over-priced'; this is a significant difference between consecutive multi-unit and unit demand models.  Formally, in a solution
$(\vec{X},\vec{p})$, we say an item $j$ is {\em over-priced} if there
is a buyer $i$ such that $j\in X_i$ and $p_j> v_i\qual_j$. That is,
the price charged for item $j$ is larger than its contribution to the
utility of its winner.

\begin{example}[Over-priced items]\label{example-over-price}
There are two buyers $i_1,i_2$ with values $v_{i_1}=20,v_{i_2}=10$ and
demands $d_{i_1}=1$ and $d_{i_2}=2$, and three items $j_1,j_2,j_3$ with
qualities $\qual_{j_1}=3,\qual_{j_2}=2,\qual_{j_3}=1$. We can see that the
allocations $X_{i_1}=\{j_1\}, X_{i_2}=\{j_2,j_3\}$ and prices
$(45,25,5)$ constitute a revenue maximizing envy-free solution with
total revenue $75$, where item $j_2$ is over-priced. If no items are
over-priced, the maximum possible prices are $(40,20,10)$ with total
revenue $70$.
\end{example}

\section{Optimal Auction for the Single Peak Case}
\label{sec:single}
The goal of this section is to present our optimal auction for the single
peak case that serves as an elementary component in the general case later.
En route, several principal techniques are examined
exhaustively to the extent that they can be applied directly in the next section.
By employing these techniques, we show that the optimal
Bayesian Incentive Compatible auction can be represented by a simple
Incentive Compatible one. Furthermore, this optimal auction can be
implemented efficiently. Let
$T_i(v_i)=\E_{\vi}[t_i(\vec{v})]$, $P_i(v_i)=\E_{\vi}[p_i(\vec{v})]$ and $\phi_i(v_i)=v_i-\frac{1-F_i(v_i)}{f_i(v_i)}$. From Myerson' work \cite{My81}, we obtain the following three lemmas.
 \begin{lemma}[From \cite{My81}]
\label{lem:monotone}
  A mechanism $M=(x,p)$ is Bayesian Incentive Compatible if and only if:\\
  a) $T_i(x)$ is monotone non-decreasing for any agent $i$.\\
  b) $P_i(v_i)=v_iT_i(v_i)-\int_{\lv_i}^{v_i}T_i(z)dz$
\end{lemma}

\begin{lemma}[From \cite{My81}]
\label{lem:revenue}
  For any BIC mechanism $M=(x,p)$, the expected revenue
  $\Ev[\sum_iP_i(v_i)]$ is equal to the virtual surplus
  $\Ev[\sum_i\phi_i(v_i)t_i(\vec{v})]$.
\end{lemma}
The following lemma is the direct result of Lemma \ref{lem:monotone} and \ref{lem:revenue}.
\begin{lemma}
  \label{lem:single}
  Suppose that $x$ is the allocation function that maximizes
  $\Ev[\phi_i(v_i)t_i(\vec{v})]$ subject to the constraints that
  $T_i(v_i)$ is monotone non-decreasing for any bidders' profile
  $\vec{v}$, any agent $i$ is assigned either $d_i$ consecutive slots
  or nothing. Suppose also that
  \begin{equation}
    \label{eq:single-price}
    p_i(\vec{v})=v_it_i(\vec{v})-\int_{\lv_i}^{v_i}t_i(v_{-i},s_i)ds_i
  \end{equation}
  Then $(x,p)$ represents an optimal mechanism for the rich media
  advertisement problem in single-peak case.
\end{lemma}

We will use dynamic programming to maximize the virtual surplus in Lemma \ref{lem:revenue}.
 Suppose all the buyers are sorted in a no-increasing order according to their virtual values. We will need the following two useful lemmas.
 Lemma~\ref{lem:con} states that all the allocated slots are consecutive.
\begin{lemma}
\label{lem:con}
  There exists an optimal allocation $x$ that maximizes
  $\sum_i\phi_i(v_i)t_i(\vec{v})$ in the single peak case, and satisfies the
  following condition. For any unassigned slot $j$, it must be that either
  $\forall j'>j$, slot $j'$ is unassigned or $\forall j'<j$, slot $j'$ is unassigned.
\end{lemma}
\begin{proof}
  We pick an arbitrary optimal allocation $x$ that maximizes the
  summation of virtual values. If $x$ satisfies the property, it is
  the desired allocation and we are done. Otherwise, we do the
  following modification on $x$. Let slot $j$ $(1<j<m)$ be the
  unassigned slot between buyers' allocated slots. Since the quality
  function are single peaked, we have $q_j\ge q_{j+1}$ or
  $q_j\ge q_{j-1}$. We only prove the lemma for the case $q_j\ge q_{j+1}$
  and the proof for the other case is symmetric. Let slot $j'>j$ be
  the leftmost assigned slot on the right side of $j$. We modify $x$
  by assigning the buyer $i$ who got the slot $j'$ the $d_i$
  consecutive slots from $j$. It is easy to check the resulting
  allocation is still feasible and optimal. Moreover, the slot $j$
  becomes assigned now. By keep doing this, we can eliminate all
  unassigned slots between buyers' allocations. Thus, the resulting
  allocation must be consecutive.
\end{proof}
Next, we prove that this consecutiveness even holds for all set $[s]
\subseteq [n]$. That is, there exists an optimal allocation that
always assigns the first $s$ buyers consecutively for all $s\in [n]$.
For convenience, we say that a slot is ``out of'' a set of buyers if the slot
is not assigned to any buyers in that set. Then the consecutiveness
can be formalized in the following lemma.

\begin{lemma}
\label{lem:subcon}
  There exists an optimal allocation $x$ in the single peak case, that
  satisfies the following condition. For any slot $j$ out of $[s]$, it
  must be either $\forall j'>j$, slot $j'$ is out of $[s]$ or $\forall
  j'<j$, slot $j'$ is out of $[s]$.
\end{lemma}
\begin{proof}
    The idea is to pick an arbitrary optimal allocation $x$ and modify
  it to the desired one. Suppose $x$ does not satisfy the property on
  a subset $[s]$. By Lemma \ref{lem:con}, there is no unassigned slots
  in the middle of allocations to set $[s]$. Then there must be a slot
  assigned to a buyer $i$ out of the set $[s]$ that separates the
  allocations to $[s]$ We use $W_i$ to denote the allocated slots of
  buyer $i$. Suppose slot $k$ is the peak. There are two cases to be considered:
  \begin{itemize}
  \item[Case 1.]{\bf $k\notin W_i$}.\\
  Let $j$ and $j'$ be the leftmost and rightmost slot in
  $W_i$ respectively. We consider two cases $q_j\ge q_{j'}$ and
  $q_j<q_{j'}$.  We only prove for the first case and the proof for
  the other case is symmetric. If $q_j\ge q_{j'}$, we find the
  leftmost slot $j_1>j'$ assigned to $[s]$ and the rightmost slot
  $j_2<j_1$ not assigned to $[s]$. In addition, let $i_1 \in [s]$ be
  the buyer that $j_1$ is assigned to and $i_2>s$ be the buyer that
  $j_2$ is assigned to. In single peak case, it is easy to check
  $q_j\ge q_{j'}$ implies that all the slots assigned to $i_2$ have
  higher quality than $i_1$'s. Thus swapping the positions of $i_1$
  and $i_2$  will always
  increase the virtual surplus, $\sum_i\phi_i(v_i)t_i(\vec{v})$. By
  keep doing this, we can eliminate all slots out of $[s]$ in the
  middle of allocation to $[s]$ and attain the desired optimal
  solution.
  \item[Case 2.]{\bf $k \in W_i$}\\
   Suppose $W_i=\{j^i_{1},j^i_{2},\cdots,j^i_{u_i}\}$ with $j^i_{1}<j^i_{2}<\cdots<j^i_{u_i}$ and there exists $1\leq e\leq u_i$ such that $k=j^i_e$. Let $a$ and $b$ be the left and right neighbour buyers of $i$ winning slots next to $W_i$. As we know $a,b \in [s]$, hence, $v_a\ge v_i$ and $v_b\ge v_i$. Let $W_a=\{j^a_{1},j^a_{2},\cdots,j^a_{u_a}\}$ and $W_b=\{j^b_{1},j^b_{2},\cdots,j^b_{u_i}\}$ denote  the allocated slots of buyer $a$ and $b$ respectively, where $j^a_{1}<j^a_{2}<\cdots<j^a_{u_a}$ and $j^b_{1}<j^b_{2}<\cdots<j^b_{u_b}$. As $k\in W_i$, then $q_{j^i_1}\ge q_{j^a_{u_a}}$ and $q_{j^i_{u_i}}\ge q_{j^b_1}$ (noting that $j^a_{u_a}$ and $j^b_{1}$ are the indices of slots with the largest qualities in $W_a$ and $W_b$ respectively). We will show that either swapping winning slots of $i$ with $a$ or with $b$ will increase the virtual surplus. To prove this, there four cases needed to be considered: (1). $u_i\geq u_a$ and $u_i\geq u_b$; (2). $u_i\geq u_a$ and $u_i< u_b$; (3). $u_i<u_a$ and $u_i\geq u_b$; (4). $u_i<u_a$ and $u_i< u_b$. We only prove the case (1) since the other cases can be proved similarly. Now, suppose $u_i\geq u_a$ and $u_i\geq u_b$, then we must have
   either (i). $\sum_{k=1}^{u_b}q_{j^i_{k}}\geq \sum_{k=1}^{u_b}q_{j^b_k}$ or (ii). $\sum_{k=1}^{u_a}q_{j^i_{u_i-k+1}}\geq \sum_{k=1}^{u_a}q_{j^a_k}$.
   Suppose (i) is not true, that is  $\sum_{k=1}^{u_b}q_{j^i_{k}}< \sum_{k=1}^{u_b}q_{j^b_k}$, if $u_b\le e$, then we have $q_{j^i_1}\le q_{j^i_{u_b}}$, as a result, $$u_b q_{j^i_1}\le \sum_{k=1}^{u_b}q_{j^i_{k}}< \sum_{k=1}^{u_b}q_{j^b_k}\le u_b q_{j^b_1}\le u_b q_{j^i_{u_i}},$$
   thus, $q_{j^i_1}<q_{j^i_{u_i}}$; otherwise $u_b>e$, then it must also hold that $q_{j^i_1}\le q_{j^i_{u_b}}$ (otherwise, for any $1\le \ell\le u_b$,
   $q_{j^i_{\ell}}\ge q_{j^i_{u_b}}\ge q_{j^b_1}$ implying that $\sum_{k=1}^{u_b}q_{j^i_{k}}\geq u_bq_{j^b_1}\ge \sum_{k=1}^{u_b}q_{j^b_k}$, contradiction), hence, for any $1\le \ell\le u_b$, $q_{j^i_{\ell}}\ge q_{j^i_{1}}$, it follows,
   $$u_b q_{j^i_1}\le \sum_{k=1}^{u_b}q_{j^i_{k}}< \sum_{k=1}^{u_b}q_{j^b_k}\le u_b q_{j^b_1}\le u_b q_{j^i_{u_i}},$$
   in both cases, it is obtained that $q_{j^i_1}<q_{j^i_{u_i}}$, therefore,
    $$\sum_{k=1}^{u_a}q_{j^i_{u_i-k+1}}> u_aq_{j^i_1}\ge \sum_{k=1}^{u_a}q_{j^a_k}$$ implying (ii) is true.
Thus, if (i) is true, by simple calculations, swapping winning slots of $i$ with $b$ will increase the virtual value (since $v_b\ge v_i$), otherwise swapping winning slots of $i$ with $a$ will increase the virtual surplus (since $v_a\ge v_i$). Then keep doing it by the method of Case 1 until eliminating all slots out of $[s]$ in the
  middle of allocation to $[s]$ and attaining the desired optimal solution.
   \end{itemize}
\end{proof}
Since the optimal solution always assigns to $[s]$ consecutively (Lemma \ref{lem:subcon}), we can boil the
allocations to $[s]$ down to an interval denoted by $[l,r]$.  Let
$g[s,l,r]$ denote the maximized value of our objective function
$\sum_i\phi_i(v_i)t_i(\vec{v})$ when we only consider first $s$ buyers
and the allocation of $s$ is exactly the interval $[l,r]$. Then we
have the following transition function.
\begin{equation}
\label{equation:dp}
g[s,l,r]=\max\left\{
\begin{array}{l}
g[s-1,l,r]\\\\
g[s-1,l,r-d_s]+\phi_s(v_s)\sum_{j=r-d_s+1}^{r}q_j\\\\
g[s-1,l+d_s,r]+\phi_s(v_s)\sum_{j=l}^{l+d_s-1}q_j
\end{array}
\right.
\end{equation}

Our summary statement is as follows.

\begin{theorem}
\label{thm:monotone}
  The mechanism that applies the allocation rule according to
  Dynamic Programming~(\ref{equation:dp}) and payment rule according to
  Equation~(\ref{eq:single-price}) is an optimal mechanism for the banner
  advertisement problem with single peak qualities.
\end{theorem}
\begin{proof}
To complete the proof, it suffices to prove that $T_i(v_i)$ is monotone
  non-decreasing. More specifically, we prove a stronger fact, that
  $t_i(v_i,\vi)$ is non-decreasing as $v_i$ increases. Given other
  buyers' bids $v_{-i}$, the monotonicity of $t_i$ is equivalent to
  $t_i(v_i,v_{-i})\le t_i(v'_i,v_{-i})$ if $v'_i>v_i$. Assuming that
  $v'_i>v_i$, the regularity of $\phi_i$ implies that $\phi_i(v_i)\le
  \phi_i(v'_i)$. If $\phi_i(v_i)=\phi_i(v'_i)$, then $t_i(v_i,v_{-i})=
  t_i(v'_i,v_{-i})$ and we are done.

  Consider the case that $\phi_i(v_i)<\phi_i(v'_i)$. Let $Q$ and $Q'$
  denote the total quantities obtained by all the other buyers except
  buyer $i$ in the mechanism when buyer $i$ bids $v_i$ and $v'_i$
  respectively.
  \begin{align*}
    \phi_i(v'_i)t_i(v'_i,v_{-i})+Q'&\ge \phi_i(v'_i)t_i(v_i,v_{-i})+Q\\
     \phi_i(v_i)t_i(v_i,v_{-i})+Q&\ge \phi_i(v_i)t_i(v'_i,v_{-i})+Q'.
  \end{align*}
  Above inequalities are due to the optimality of allocations when $i$
  bids $v_i$ and $v'_i$ respectively. It follows that
  \begin{align*}
    \phi_i(v'_i)(t_i(v_i,v_{-i})-t_i(v'_i,v_{-i}))\le Q'-Q\\
    \phi_i(v_i)(t_i(v_i,v_{-i})-t_i(v'_i,v_{-i})) \ge Q'-Q
  \end{align*}
  By the fact that $\phi_i(v_i)<\phi_i(v'_i)$, it must be
  $t_i(v_i,v_{-i})\le t_i(v'_i,v_{-i})$.
\end{proof}

\section{Multiple Peaks  Case}
\label{sec:general}

Suppose now that there are only $h$ peaks (local maxima) in
the qualities. Thus, there are at most $h-1$ valleys (local
minima). Since $h$ is a constant, we can enumerate all the buyers
occupying the valleys. After this enumeration, we can divide the
qualities into at most $h$ consecutive pieces and each of them forms a
single-peak. Then using similar properties as those in Lemma \ref{lem:con} and \ref{lem:subcon},
 we can obtain a larger size dynamic programming (still runs in polynomial time) similar
 to dynamic programming (\ref{equation:dp}) to solve the problem.

\begin{theorem}
\label{thm:multipeak}
  There is a polynomial algorithm to compute revenue maximization
  problem in Bayesian settings where the qualities of slots have a
  constant number of peaks.
\end{theorem}
\begin{proof}
  Our proof is based on the single peak algorithm. Assume there are $h$
  peaks, then there must be $h-1$ valleys. Suppose these valleys are
  indexed $j_1,j_2,\cdots,j_{h-1}$. In optimal allocation, for any
  $j_k$, $k=1,2,\cdots,h-1$, $j_k$ must be allocated to a buyer or
  unassigned to any buyer. If $j_k$ is assigned to a buyer, say, buyer
  $i$, since $i$ would buy $d_i$ consecutive slots, $j_k$ may appear
  in $\ell$th position of this $d_i$ consecutive slots. Hence, by this
  brute force, each $j_k$ will at most have $\sum_i d_i+1\le mn+1$
  possible positions to be allocated. In all, all the valleys have
  $(mn+1)^h$ possible allocated positions. For each of this
  allocation, the slots is broken into $h$ single peak slots. We can obtain similar properties as those in Lemma \ref{lem:con} and \ref{lem:subcon}.
  Without loss of generality, suppose the rest buyers are still the set $[n]$, with non-increasing virtual value.
  Since the optimal solution always assigns to $[s]$ consecutively, we can boil the
allocations to $[s]$ down to  intervals denoted by $[l_i,r_i]$, $i=1,2,\cdots,d$, where $[l_i,r_i]$ lies in the $i$-th single peak slot.  Let
$g[s,l_1,r_1,\cdots,l_d,r_d]$ denote the maximized value of our objective function
$\sum_i\phi_i(v_i)t_i(\vec{v})$ when we only consider first $s$ buyers
and the allocations of $[s]$ are exactly intervals $[l_i,r_i]$, $i=1,2,\cdots,d$. Then we
have the following transition function.
\begin{equation*}
g[s,l_1,r_1,\cdots,l_d,r_d]=\max_{i\in[d]}\left\{
\begin{array}{l}
g[s-1,l_1,r_1,\cdots,l_d,r_d]\\\\
g[s-1,l_1,r_1,\cdots,l_i,r_i-d_s,\cdots,l_d,r_d]+\phi_s(v_s)\sum_{j=r_i-d_s+1}^{r_i}q_j\\\\
g[s-1,l_1,r_1,\cdots,l_i+d_s,r_i,\cdots,l_d,r_d]+\phi_s(v_s)\sum_{j=l_i}^{l_i+d_s-1}q_j
\end{array}
\right.
\end{equation*}
\end{proof}

Now we consider the case without the constant peak assumption and
prove the following hardness result.

\begin{theorem}(NP-Hardness)
\label{thm:hard}
  The revenue maximization problem for rich media ads with arbitrary
  qualities is NP-hard.
\end{theorem}
\begin{proof}
 We prove the NP-hardness by reducing the $3$ partition problem that
  is to decide whether a given multi-set of integers can be partitioned
  into certain number of subsets that all have the same sum. More precisely, given a
  multi-set $S$ of $3n$ positive integers, can $S$ be partitioned
  into $n$ subsets $S_1,\ldots, S_n$ such that the sum of the numbers
  in each subset is equal? The $3$ partition problem has been proven
  to be NP-complete in a strong sense in \cite{GJ75}, meaning that it
  remains NP-complete even when the integers in $S$ are bounded above
  by a polynomial in $n$.

  Given a instance of $3$ partition $(a_1,a_2,\ldots,a_{3n})$, we
  construct a instance for advertising problem with $3n$ advertisers
  and $m=n+\sum_ia_i$ slots. It should be mentioned that $m$ is
  polynomial of $n$ due to the fact that all $a_i$ are bouned by a
  polynomial of $n$. In the advertising instance, the valuation $v_i$
  for each advertiser $i$ is $1$ and his demand $d_i$ is defined as
  $a_i$. Moreover, for any advertiser, his valuation distribution is
  that $v_i=1$ with probability $1$. Then everyone's virtual value is
  exactly $1$. To maximize revenue is
  equivalent to maximize the simplified function
  $\sum_i\sum_jx_{ij}q_j$.

  Let $B=\sum_ia_i/n$. We define the quality of slot $j$ is $0$ if $j$
  is times of $B+1$, otherwise $q_j=1$. That can be illustrated as
  follows.
  \begin{displaymath}
    \underbrace{1\,1\cdots 1}_{B}\,0\,\underbrace{1\,1\cdots 1}_{B}\,0 \ldots\underbrace{1\,1\cdots 1}_{B}\,0
  \end{displaymath}

  It is not hard to see that the optimal revenue is $\sum_ia_i$ iff
  there is a solution to this $3$ partition instance.
\end{proof}

\section{Competitive Equilibrium}
\label{sec:eq}
In this section, we study the revenue maximizing competitive
equilibrium in the full information setting. To simplify the following
discussions, we sort all buyers and items in non-increasing order of
their values, i.e., $v_1\geq v_2\geq\cdots\geq v_n$.

We say an allocation $\vec{Y}=(Y_1,Y_2,\cdots,Y_n)$ is efficient if
$\vec{Y}$ maximizes the total social welfare e.g. $\sum_i\sum_{j\in
  Y_i}v_{ij}$ is maximized over all the possible allocations. We call
$\vec{p}=(p_1,p_2,\cdots,p_m)$ an equilibrium price if there exists
an allocation $\vec{X}$ such that $(\vec{X},\vec{p})$ is a
competitive equilibrium. The following lemma is implicitly stated in~\cite{GS99},
for completeness, we give a proof below.
\begin{lemma}\label{lem:CE}
  Let allocation $\vec{Y}$ be efficient, then for any equilibrium
  price $\vec{p}$, $(\vec{Y},\vec{p})$ is a competitive
  equilibrium.
\end{lemma}
\begin{proof}
  Since $\vec{p}$ is an equilibrium price, there exists an
  allocation $\vec{X}$ such that $(\vec{X},\vec{p})$ is a
  competitive equilibrium. As a result, by envy-freeness,
  $u_i(\vec{X},\vec{p})\ge u_i(\vec{Y},\vec{p})$ for any
  $i\in [n]$. Let $T=[m]\backslash \cup_iY_i$, then we have
\begin{eqnarray}
&&\sum_i\sum_{j\in Y_i}v_{ij}-\sum_{j=1}^mp_j \ge \sum_i\sum_{j\in X_i}v_{ij}-\sum_{j=1}^mp_j=\sum_i\sum_{j\in X_i}v_{ij}-\sum_i\sum_{j\in X_i}p_j\nonumber\\
&=&\sum_iu_i(\vec{X},\vec{p})\ge \sum_iu_i(\vec{Y},\vec{p})=\sum_i\sum_{j\in Y_i}v_{ij}-\sum_i\sum_{j\in Y_i}p_j\nonumber\\
&=&\sum_i\sum_{j\in Y_i}v_{ij}-\sum_{j=1}^mp_j+\sum_{j\in T}p_j
\end{eqnarray}
where the first inequality is due to $\vec{Y}$ being efficient and
first equality  due to $u_i(\vec{X},\vec{p})$ being competitive
equilibrium (unallocated item priced at $0$). Therefore, $\sum_{j\in
 T}p_j=0$ and the above inequalities are all equalities. $\forall i:
u_i(\vec{X},\vec{p})= u_i(\vec{Y},\vec{p})$. Further,
because the price is the same,

$\forall i$ a loser $\forall Z$ consecutive items and $|Z|=d_i$, we have $u_i(Z)\leq 0$.

$\forall i$ a winner $\forall Z$ consecutive items and $|Z|=d_i$, we have
$$u_i(Y_i)=u_i(X_i)\geq u_i(Z).$$
Therefore, $(\vec{Y},\vec{p})$ is a competitive equilibrium.
\end{proof}
By Lemma~\ref{lem:CE}, to find a revenue maximizing competitive
equilibrium, we can first find an efficient allocation and then use
linear programming to settle the prices.  We develop the following
dynamic programming to find an efficient allocation.  We first only
consider there is one peak in the quality order of items.  The case
with constant peaks is similar to the above approaches, for general
peak case, as shown in above Theorem~\ref{thm:hard}, finding one
competitive equilibrium is NP-hard if the competitive equilibrium
exists, and determining existence of competitive equilibrium is also
NP-hard. This is because that considering the instance in the proof of
Theorem~\ref{thm:hard}, it is not difficult to see the constructed
instance has an equilibrium if and only if 3 partition has a
solution.

Recall that all the values are sorted in non-increasing order e.g. $v_1\ge v_2
\ge\cdots\ge v_n$. $g[s,l,r]$ denotes the maximized value of social
welfare when we only consider first $s$ buyers and the allocation of
$s$ is exactly the interval $[l,r]$. Then we have the following transition function.
\begin{equation}
\label{eq:CE}
g[s,l,r]=\max\left\{
\begin{array}{l}
g[s-1,l,r]\\\\
g[s-1,l,r-d_s]+v_s\sum_{j=r-d_s+1}^{r}q_j\\\\
g[s-1,l+d_s,r]+v_s\sum_{j=l}^{l+d_s-1}q_j
\end{array}
\right.
\end{equation}
By tracking procedure~\ref{eq:CE}, an efficient allocation denoted by
$\vec{X}^*=(X^*_1,X^*_2,\cdots,X^*_n)$ can be found. The price
$\vec{p}^*$ such that $(\vec{X}^*,\vec{p}^*)$ is a revenue maximization
competitive equilibrium can be determined from the following linear
programming. Let $T_i$ be any consecutive number of $d_i$ slots, for
all $i\in [n]$.
\begin{align*}
  \max \quad & \sum_{i\in [n]}\sum_{j\in X^*_i}p_j &  \\
  s.t. \quad & p_j \ge 0 & \forall\ j\in [m] \\
  & p_j = 0 & \forall\ j\notin \cup_{i\in [n]}X^*_i   \\
  & \sum_{j\in X^*_i}(v_i\qual_j-p_j)\ge \sum_{j'\in T_i}(v_i\qual_{j'}-p_{j'}) & \forall\ i\in [n] \\
  & \sum_{j\in X^*_i}(v_i\qual_j-p_j)\ge 0 &\forall i\in[n]
\end{align*}

Clearly there is only a polynomial number of constraints. The
constraints in the first line represent that all the prices are non negative
(no positive transfers). The constraint in the second line means
unallocated items must be priced at zero (market clearance condition).
And the constraint in the third line contains two aspects of information.
First for all the losers e.g. loser $k$ with
$X_k=\emptyset$, the utility that $k$ gets from any consecutive number of
$d_k$ is no more than zero, which makes all the losers envy-free.
The second aspect is that the winners e.g. winner $i$ with $X_i\neq
\emptyset$ must receive a bundle with $d_i$ consecutive slots
maximizing its utility over all $d_i$ consecutive slots, which
together with the constraint in the fourth line (winner's utilities are non
negative) guarantees that all winners are envy-free.
\begin{theorem}\label{Thm-CE}
  Under the condition of a constant number of peaks in the qualities of slots,
  there is a polynomial time algorithm to decide whether there exists
  a competitive equilibrium or not and to compute a revenue maximizing
  revenue market equilibrium if one does exist.
  If the number of peaks in the qualities of the slots is unbounded, both the problems
  are NP-complete.
\end{theorem}
\begin{proof}
  Clearly the above linear programming and procedure (\ref{eq:CE}) run
  in polynomial time. If the linear programming output a price
  $\vec{p}^*$, then by its constraint conditions,
  $(\vec{X}^*,\vec{p}^*)$ must be a competitive equilibrium. On the
  other hand, if there exist a competitive equilibrium
  $(\vec{X},\vec{p})$ then by Lemma~\ref{lem:CE},
  $(\vec{X}^*,\vec{p})$ is a competitive equilibrium, providing a
  feasible solution of above linear programming. By the objective of
  the linear programming, we know it must be a revenue maximizing one.
\end{proof}
$\quad$

\section{Consecutive Envy-freeness}\label{sec:envy}
We first prove a negative result on computing the revenue maximization problem in general demand case. We show it is NP-hard even if all the qualities are the same.
\begin{theorem}\label{Thm-cons-2}
The revenue maximization problem of consecutive envy-free buyers is NP-hard even if all the qualities are the same.
\end{theorem}
\begin{proof}
 We prove the NP-hardness by reducing the $3$ partition problem that
  is to decide whether a given multi-set of integers can be partitioned
  into certain number of subsets that all have the same sum. More precisely, given a
  multi-set $S$ of $3n$ positive integers, can $S$ be partitioned
  into $n$ subsets $S_1,\ldots, S_n$ such that the sum of the numbers
  in each subset is equal? The $3$ partition problem has been proven
  to be NP-complete in a strong sense in \cite{GJ75}, meaning that it
  remains NP-complete even when the integers in $S$ are bounded above
  by a polynomial in $n$.

  Given a instance of $3$ partition $(a_1,a_2,\ldots,a_{3n})$. Let $B=\sum_ia_i/n$. we
  construct a instance for advertising problem with $3n+1$ advertisers
  and $m=B+1+n+\sum_ia_i$ slots. It should be mentioned that $m$ is
  polynomial of $n$ due to the fact that all $a_i$ are bounded by a
  polynomial of $n$. In the advertising instance, the valuation $v_i$
  for each advertiser $i$ is $1$ and his demand $d_i$ is defined as
  $a_i$ and there is another buyer with valuation $2$ for each slot and with demand $B+1$.
  The quality of each slot $j$ is $1$.   It is not hard to see that the optimal revenue is $nB+2(B+1)$ if and only if
  there is a solution to this $3$ partition instance, the optimal solution is illustrated as follows.
\begin{displaymath}
    \underbrace{1\,1\cdots 1}_{B+1}\,\underbrace{1}_{\text{unassigned}}\,\underbrace{1\,1\cdots 1}_{B}\,\underbrace{1}_{\text{unassigned}}\,\underbrace{1\,1\cdots 1}_{B}\,\underbrace{1}_{\text{unassigned}} \ldots\underbrace{1\,1\cdots 1}_{B}\
\end{displaymath}
\end{proof}
Although the hardness in Theorem~\ref{Thm-cons-2} indicates that finding the optimal revenue for general demand in polynomial time is impossible , however, it doesn't rule out the very important case where the demand is uniform, e.g. $d_i=d$. We assume slots are in a decreasing order from top to bottom, that is, $q_1\geq q_2\ge \cdots \ge q_m$ . The result is summarized as follows.
\begin{theorem}\label{Thm-cons-3}
There is a polynomial time algorithm to compute the consecutive envy-free solution when all the buyers have the same demand and slots are ordered from top to bottom.
\end{theorem}
The proof of Theorem~\ref{Thm-cons-3} is based on bundle envy-free solutions, in fact we will prove the bundle envy-free solution is also a consecutive envy-free solution by defining price of items properly. Thus, we need first give the result on bundle envy-free solutions.
Suppose $d$ is the uniform demand for all the buyers. Let $T_i$ be the slot set allocated to buyer $i$, $i=1,2,\cdots,n$. Let $P_i$ be the total payment of buyer $i$ and $p_j$ be the price of slot $j$.
Let $t_i$ denote the total qualities obtained by buyer $i$, e.g. $t_i=\sum_{j\in T_i}q_j$ and $\alpha_i=iv_i-(i-1)v_{i-1}$, $\forall i\in [n]$.
\begin{theorem}\label{Thm-bundle-EF}
 The revenue maximization problem of bundle envy-freeness is equivalent to solving the following LP.
\begin{equation}\label{EF-upper}
\begin{aligned}
\textrm{Maximize:}\quad&\sum_{i=1}^n\alpha_it_i &\\
\textrm{s.t.}\quad & t_1\ge t_2\ge\cdots\ge t_n &\\
&T_i\subset [m], \ \ T_i\cap T_k=\emptyset \ \ \forall i,k\in [n]&
\end{aligned}
\end{equation}
\end{theorem}
\begin{proof}[of Theorem \ref{Thm-bundle-EF}]
Recall $P_i$ denote the payment of buyer $i$, we next prove that the linear programming (\ref{EF-upper}) actually gives optimal solution of bundle envy-free. By the definition of bundle envy-free,
where buyer $i$ would not envy buyer $i+1$ and versus, we have
\begin{equation}\label{EF-ine-1}
v_it_i-P_i\ge v_it_{i+1}-P_{i+1}
\end{equation}
\begin{equation}\label{EF-ine-2}
v_{i+1}t_{i+1}-P_{i+1}\ge v_{i+1}t_i-P_i
\end{equation}
Plus above two inequalities gives us that $(v_i-v_{i+1})(t_i-t_{i+1})\ge 0$. Hence, if $v_i>v_{i+1}$, then $t_i\ge t_{i+1}$.
From (\ref{EF-ine-1}), we could get $P_i\le v_i(t_i-t_{i+1})+P_{i+1}$.
The maximum payment of buyer $i$ is
\begin{equation}\label{EF-ine-3}
P_i= v_i(t_i-t_{i+1})+P_{i+1},
\end{equation}
 together with $t_i\ge t_{i+1}$, implying (\ref{EF-ine-1}) and (\ref{EF-ine-2}). Besides the maximum payment of $n$ is $P_n=t_nv_n$. (\ref{EF-ine-3}) together with $t_i\ge t_{i+1}$ and $P_n=t_nv_n$ would make everyone bundle envy-free, the arguments are as follows.
\begin{itemize}
 \item All the buyers must be bundle envy free. By (\ref{EF-ine-3}), we have $P_i-P_{i+1}= v_i(t_i-t_{i+1})$, hence $P_i=\sum_{k=i}^{n-1}v_k(t_k-t_{k+1})+P_n$. Noticing that if $t_i=0$, then $P_i=0$, which means $i$ is loser. For any buyer $j<i$, we have $P_j-P_i=\sum_{k=j}^{i-1}v_k(t_k-t_{k+1})\le\sum_{k=j}^{i-1}v_j(t_k-t_{k+1})=v_j(t_j-t_i)$. rewrite $P_j-P_i\le v_j(t_j-t_i)$ as $v_jt_i-P_i\le v_jt_j-P_j$, which means buyer $j$ would not envy buyer $i$. Similarly, $P_j-P_i=\sum_{k=j}^{i-1}v_k(t_k-t_{k+1})\ge \sum_{k=j}^{i-1}v_i(t_k-t_{k+1})=v_i(t_j-t_i)$, rewrite $P_j-P_i\ge v_i(t_j-t_i)$ as $v_it_i-P_i\ge v_it_j-P_j$, which means $i$ would not envy buyer $j$.
\end{itemize}
Now let's calculate $\sum_{i=1}^nP_i$ based on (\ref{EF-ine-3}) using notation $t_{n+1}=0$,  one has
\begin{displaymath}
\begin{split}
\sum_{i=1}^nP_i&=\sum_{i=1}^n\left[\sum_{k=i}^{n-1}v_k(t_k-t_{k+1})+P_n\right]=\sum_{i=1}^n\sum_{k=i}^{n}v_k(t_k-t_{k+1})\\
 &=\sum_{k=1}^n\sum_{i=1}^kv_k(t_k-t_{k+1})=\sum_{k=1}^nkv_k(t_k-t_{k+1})\\
 &=\sum_{k=1}^nkv_kt_k-\sum_{k=1}^n(k-1)v_{k-1}t_k=\sum_{i=1}^n\alpha_it_i
\end{split}
\end{displaymath}
We know the  revenue maximizing problem of bundle envy-freeness  can be  formalized as (\ref{EF-upper}).
\end{proof}
Since consecutive envy-free solutions are a subset of (sharp) bundle envy-free solutions, hence the optimal value of optimization (\ref{EF-upper}) gives an upper bound of optimal objective value of  consecutive envy-free solutions. Noting optimization LP (\ref{EF-upper}) can be solved by dynamic programming. Let $g[s,j]$ denote the optimal objective value of the following LP with some set in $[j]$ allocated to all the buyers in $[s]$:
\begin{equation}\label{EF-sub-bundle}
\begin{aligned}
\textrm{Maximize:}\quad&\sum_{i=1}^s\alpha_it_i &\\
\textrm{s.t.}\quad & t_1\ge t_2\ge\cdots\ge t_s &\\
&T_i\subset [j], \ \ T_i\cap T_k=\emptyset \ \ \forall i,k\in [s]&
\end{aligned}
\end{equation}
Then
\begin{equation*}\label{eq:bundle}
g[s,j]=\max\left\{
\begin{array}{l}
g[s,j-1]\\\\
g[s-1,j-d]+\alpha_s\sum_{u=j-d+1}^{j}q_u
\end{array}
\right.
\end{equation*}
Next, we show how to modify the (sharp) bundle envy-free solution to consecutive envy-free solutions by properly defining the slot price of $T_i$, for all $i\in[n]$.  Suppose the optimal winner set of bundle envy-free solution is $[L]$. Assume the optimal allocation and price of bundle envy-free solution are $T_i=\{j^i_1,j^i_2,\cdots,j^i_d\}$ with $j^i_1\ge j^i_2\ge\cdots\ge j^i_d$ and  $P_i$ respectively, for all $i\in [L]$.
\begin{proof}[ of Theorem~\ref{Thm-cons-3}]
Define the price of $T_i$ iteratively as follows:\\
$p_{j^L_k}=v_Lq_{j^L_k}$, for all $k\in [d]$;\\
$p_{j^i_k}=v_i(q_{j^i_k}-q_{j^{i+1}_k})+p_{j^{i+1}_k}$ for $k\in [d]$ and $i\in [n]$\\
Now we could see that the price defined by above procedure is still a bundle envy-free solution. This is because
by induction, we have  $P_i=\sum_{k=1}^dp_{j^i_k}$. Hence, we need only to check the prices defined as above and allocations $T_i$  constitute a consecutive envy-free solution. In fact, we prove a strong version, suppose $T_i$s are consecutive from top to down in a line $S$, we will show each buyer $i$ would not envy any consecutive sub line of $S$ comprising  $d$ slots. For any $i$,\\
\textbf{Case 1},  buyer $i$ would not envy the slots below his slots. \\
for any consecutive line $T$ below $i$ with size $d$, suppose $T$
comprises of slots won by buyer $k$ (denoted such slot set by $U_k$)
and $k+1$ (denoted such slot set by $U_{k+1}$ and let $\ell=|U_{k+1}|$) where $k\ge i$. Recall that $t_i=\sum_{j\in T_i}q_j$, then
\begin{eqnarray}
&&\sum_{j\in T_i}p_j-\sum_{j\in T}p_j=v_i(t_i-t_{i+1})+P_{i+1}-\sum_{j\in U_k\cup U_{k+1}}p_j\nonumber\\
&=&v_i(t_i-t_{i+1})+v_{i+1}(t_{i+1}-t_{i+2})+\cdots+P_k-\sum_{j\in U_k\cup U_{k+1}}p_j\nonumber\\
&=&v_i(t_i-t_{i+1})+v_{i+1}(t_{i+1}-t_{i+2})+\cdots+\sum_{j\in T_k\backslash U_k}p_j-\sum_{j\in U_{k+1}}p_j\nonumber\\
&=&v_i(t_i-t_{i+1})+v_{i+1}(t_{i+1}-t_{i+2})+\cdots+\sum_{u=1}^{\ell}v_k(q_{j^k_u}-q_{j^{k+1}_u})\nonumber\\
&\le& v_i(t_i-t_{i+1})+v_{i}(t_{i+1}-t_{i+2})+\cdots+\sum_{u=1}^{\ell}v_i(q_{j^k_u}-q_{j^{k+1}_u})\nonumber\\
&=&v_it_i-v_i\sum_{j\in T}q_j.
\end{eqnarray}

 Rewrite $\sum_{j\in T_i}p_j-\sum_{j\in T}p_j\le v_it_i-v_i\sum_{j\in T}q_j$\\
as \quad\quad\quad$v_it_i-\sum_{j\in T_i}p_j\ge v_i\sum_{j\in T}q_j-\sum_{j\in T}p_j$\\
we get the desired result.\\
\textbf{Case 2}, buyer $i$ would not envy the slots above his slots.\\
 for any consecutive line $T$ above $i$ with size  $d$, suppose $T$
 comprises of slots won by buyer $k$ (denoted such slot set by $U_k$)
 and $k-1$ (denoted such slot set by $U_{k-1}$ and let $\ell=|U_{k-1}|$) where $k\le i$. Recall that $t_i=\sum_{j\in T_i}q_j$, then
\begin{eqnarray}
&&\sum_{j\in T}p_j-\sum_{j\in T_i}p_j=\sum_{j\in U_{k-1}\cup U_{k}}p_j-\sum_{j\in T_i}p_j\nonumber\\
&=&\sum_{u=d-\ell+1}^dv_{k-1}(q_{j^{k-1}_u}-q_{j^k_u})+\sum_{j\in T_k}p_j-\sum_{j\in T_i}p_j\nonumber\\
&=&\sum_{u=d-\ell+1}^dv_{k-1}(q_{j^{k-1}_u}-q_{j^k_u})+v_k(t_k-t_{k+1})+\cdots+v_{i-1}(t_{i-1}-t_i)\nonumber\\
&\ge& \sum_{u=d-\ell+1}^dv_{i}(q_{j^{k-1}_u}-q_{j^k_u})+v_i(t_k-t_{k+1})+\cdots+v_{i}(t_{i-1}-t_i)\nonumber\\
&=&v_i\sum_{j\in T}q_j-v_it_i.
\end{eqnarray}
 Rewrite $\sum_{j\in T}p_j-\sum_{j\in T_i}p_j\ge v_i\sum_{j\in T}q_j -v_it_i$\\
as \quad\quad\quad$v_it_i-\sum_{j\in T_i}p_j\ge v_i\sum_{j\in T}q_j-\sum_{j\in T}p_j$\\
we get the desired result.\\
\end{proof}
\section{Simulation}\label{sec:simulation}
Since the consecutive model has a direct application for rich media advertisement, the simulation for comparing the schemes e.g. Bayesian optimal mechanism (Bayesian for simplicity in this chapter), consecutive CE (CE for simplicity in this chapter), consecutive EF (EF for simplicity in this chapter), generalized GSP, will be presented in this chapter. Our simulation shows a convincing result for these schemes. We did a simulation to compare the expected revenue among those pricing schemes. The sampling method is applied to the competitive equilibrium, envy-free solution, Bayesian truthful mechanism, as well as the generalized GSP, which is the widely used pricing scheme for text ads in most advertisement platforms nowadays.

The value samples $v$ come from the same uniform distribution $U[20,80]$. With a random number generator, we produced 200 group samples $\{V_1,V_2,\cdots,V_{200}\}$, they will be used as the input of our simulation. Each group contains $n$ samples, e.g. $V_k=\{v_k^1,v_k^2,\cdots,v_k^n\}$, where each $v^i_k$ is sampled from uniform distribution $U[20,80]$.  For the parameters of slots, we assume there are 6 slots to be sold, and fix their position qualities:
\begin{eqnarray}
Q&=&\{q_1,~~q_2,~~q_3,~~q_4,~~q_5,~~q_6\}\nonumber\\
&=&\{0.8,~0.7,~0.6,~0.5,~0.4,~0.3\}\nonumber\\
\end{eqnarray}

The actual ads auction is complicated, but we simplified it in our simulation, we do not consider richer conditions, such as set all bidders' budgets unlimited, and there is no reserve prices in all mechanisms. We vary the group size $n$ from 5 to 12, and observe their expected revenue variation. From $j=1$ to $j=200$, at each $j$, invoke the function EF ($V_j,D,Q$), GSP ($V_j,D,Q$), CE ($V_j,D,Q$) and Bayesian ($V_j,D,Q$) respectively. Thus, those mechanisms use the same data from the same distribution as inputs and compare their expected revenue fairly. Finally, we average those results from different mechanisms respectively, and compare their expected revenue at sample size $n$.

The generalized GSP mechanism for rich ads in the simulation was not introduced in the previous sections. Here, in our simulation, it is  a simple generalization of the standard GSP which is used in keywords auction. In our generalization of GSP, the allocations of bidders are given by maximizing the total social welfare, which is compatible with GSP in keywords auction, and each winner's price per quality is given by the next highest bidder's bid per quality. Since the real generalization of GSP for rich ads is unknown and the generalization form may be various, our generalization of GSP for rich ads may not be a revenue maximizing one, however, it is a natural one.  The pseudo-codes are listed in Appendix \ref{app:code}.

Incentive analysis is also considered in our simulation, except Bayesian mechanism (it is truthful bidding, $b_i=v_i$). Since the bidding strategies in other mechanisms (GSP, CE, EF) are unclear, we present a simple bidding strategy for bidders to converge to an equilibrium. We try to find the equilibrium bids by searching the bidder's possible bids ($b_i<v_i$) one by one, from top rank bidders to lower rank bidders iteratively, until reaching an equilibrium where no one would like to change his bid. If any equilibrium exists, we count the expected revenue for this sample; if not, we ignore this sample. All the pseudo-codes are listed in Appendix \ref{app:code}.

Since the Envy-Free solution in our paper only works for the condition that all the bidders have the same demand, thus, we did the simulation in 2 separate ways:
\begin{enumerate}
\item Simulation I is for bidders with a fixed demands, we set $d_i=2$, for all $i$ and compares expected revenues obtained by GSP, CE, EF, Bayesian.
\item Simulation II is for bidders with different demands and compares expected revenues obtained by GSP, CE, Bayesian.
\end{enumerate}
\begin{figure}[!h]
\centering
\includegraphics[width=.75\textwidth]{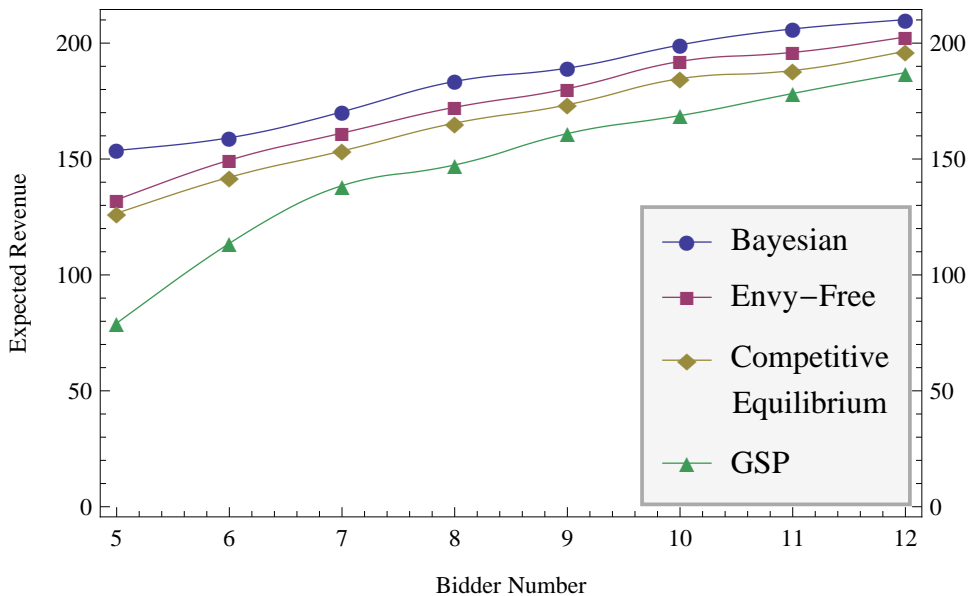}
\caption{Simulation results from different mechanisms, all bidders' demand fixed at $d_i=2$}
\label{fig:simu1}
\end{figure}

Figure \ref{fig:simu1} shows  I's results when all bidders' demand fixed at 2. Obviously, the expected revenue is increasing when more bidders involved. When the bidders' number rises, the rank of expected revenue of different mechanisms remains the same in the order Bayesian $>$ EF $>$ CE $>$ GSP.

Simulation II is for bidders with various demands. With loss of generality, we assume that bidder's demand $D=\{d_1,d_2,\cdots,d_i\},d_i\in\{1,2,3\}$, we assign those bidders' demand randomly, with equal probability.
\begin{figure}[!h]
\centering
\includegraphics[width=.75\textwidth]{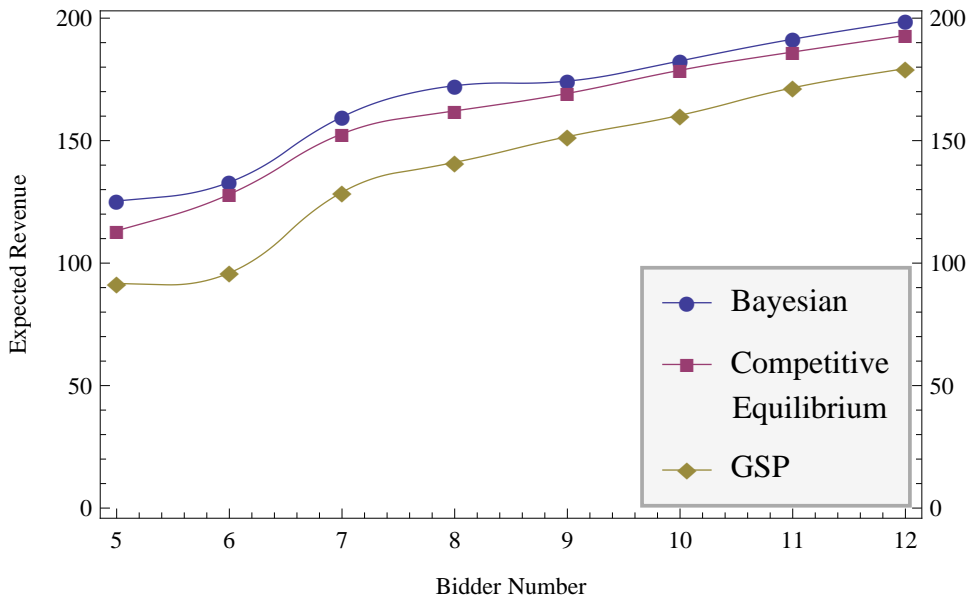}
\caption{Simulation results from different mechanisms, bidders' demand varies in \{1,2,3\}}
\label{fig:simu2}
\end{figure}

Figure \ref{fig:simu2} shows our simulation results for II when bidders' demand varies in \{1,2,3\}, the rank of expected revenue of different mechanisms remains the same as simulation I, From this chart, we can see that Bayesian truthful mechanism and competitive equilibrium get more revenues than generalized GSP.

\section{Conclusion and Discussion}

The rich media pricing models for consecutive demand buyers in the context
of Bayesian truthfulness, competitive equilibrium and envy-free solution paradigm are investigated in this paper.
As a result, an optimal Bayesian incentive compatible  mechanism is proposed for various
settings such as single peak and multiple  peaks.
In addition, to incorporate fairness e.g. envy-freeness, we also present a polynomial-time
algorithm to decide whether or not there exists a competitive equilibrium or and to compute
a revenue maximized market equilibrium if one does exist. For envy-free settings, though the revenue maximization of general demand case is shown to be NP-hard, we still provide optimal solution of common demand case. Besides, our simulation shows a reasonable relationship of revenues  among these schemes plus a generalized GSP for rich media ads.

Even though our main motivation arises from the rich media advert pricing problem,
our models have other potential applications. For example TV ads can also be modeled
under our consecutive demand adverts where inventories of a commercial break
are usually divided into slots of fixed sizes, and slots have various qualities measuring their
expected number of viewers and corresponding attractiveness.
With an extra effort to explore the periodicity of TV ads, we can extend our multiple peak model to
one involved with cyclic multiple peaks.
Besides single consecutive demand where each buyer only have one demand choice,
the buyer may have more options to display his ads, for example select a large picture
or a small one to display them. Our dynamic programming algorithm (\ref{equation:dp}) can also
be applied to this case (the transition function in each step selects maximum value
from $2k+1$ possible values, where $k$ is the number of choices of the buyer).

Another reasonable extension of our model would be to add budget constraints for buyers,
i.e., each buyer cannot afford the payment more than his budget.
By relaxing the requirement of Bayesian incentive compatible (BIC) to one of approximate BIC,
this extension can be obtained by the recent milestone work of Cai et al.~\cite{CDW12-2}.
It remains an open problem how to do it under the exact BIC requirement.
It would also be interesting to handle it under the market equilibrium paradigm for our model.

\newpage
\section*{Appendix}

\section{Pseudo-code of Simulation}\label{app:code}
\subsection{Expected Revenue for Bayesian Truthful Mechanism}
Suppose with loss of generality, $b_1>b_2>\ldots>b_n>10$, and $q_1>q_2>\ldots>q_n$, let $\phi_i(v_i)=2v_i-b_i-10$.
\begin{algorithm}[!h]
\SetAlgoLined
\caption{Bayesian Expected Revenue}
\label{alg:BER}
\KwIn{Demands $d_i$, qualities(CTR) $q_j$ and bids $b_i$, number of samples $K$}
\KwOut{Expected Revenue R}
Generate uniform distribution for $b_i$ as $\mathbf{I}_i$ uniformly distributed on $I_i=[b_i-10,b_i+10]$\;
Repeat \;

\For {$r=1,2,\cdots,K$}{
  Generate $v^r_i$ from $\mathbf{I}_i$ independently, $i=1,2,\cdots,n$\;
  Calculate $\phi_i(v^r_i)$ and sort it decreasing  order as $\phi'_i(v^r_i)>\phi'_{i+1}(v^r_i)$, $i=1,2,\cdots,n$\;
  Use dynamic programming
        \begin{equation}
\label{eq:dp}
g[s,r]=\max\left\{
\begin{array}{l}
g[s-1,r]\\\\
g[s-1,r-d_s]+\phi'_s(v^r_s)\sum_{j=r-d_s+1}^{r}q_j\\
\end{array}
\right.
\end{equation}
By tracking dynamic programming find allocation $X_i$;\\
Calculate $R^r=\sum_i\phi_i(v_i^r)\sum_{j\in X_i}q_j$\\
}
return $R=\frac{1}{K}\sum_{r=1}^KR^r$\;
\end{algorithm}

\newpage

The following the sub algorithm for finding the allocations $X_i$ when $\phi_i$, $i=1,2,\cdots,n$ are known.
\begin{algorithm}[!h]
\SetAlgoLined
\caption{sharp}
\label{alg:sharp}
\KwIn{virtual surplus $\phi_i$ qualities $q_j$ }
\KwOut{Allocation $x_{ij}$}
Sort buyers $i$ in decreasing order of $\phi_i$\;
$g[i,j]\leftarrow -\infty$; $g[0,0]\leftarrow 0$\;
$u[i,j]\leftarrow 0$; $x_{ij}\leftarrow 0$\;
\For {each buyer $i$ with positive $\phi_i$}{
  \For {each item $j$}{
    $tmp \leftarrow g[i-1,j-d_i]+\sum_{k=j-d_i+1}^{j}\phi_iq_k$\;
    $g[i,j]\leftarrow g[i-1,j]$\;
    \If{$g[i,j]<tmp$}{
      $u[i,j]\leftarrow 1$\;
      $g[i,j]\leftarrow tmp$\;
    }
  }
}
$g[i^*,j^*]=\max_{i,j}\{g[i,j]\}$\;
\While {$i^*>0$}{
  \If {$u[i^*,j^*]=1$}{
    \For {each item $k$ from $j^*-d_{i^*}+1$ to $j^*$}{
      $x_{i^*,k}\leftarrow 1$\;
    }
    $j^*\leftarrow j^*-d_{i^*}$\;
  }
  $i^*\leftarrow i^*-1$\;
}
return $x$\;
\end{algorithm}

\newpage

\subsection{Revenue from Competitive Equilibrium}
Suppose $q_1\ge q_2\ge q_3\ge \cdots\ge q_n$
\begin{algorithm}[!h]
\SetAlgoLined
\caption{Sub-algorithm for CE denoted by CE(d,q,b)}
\label{alg:CE}
\KwIn{Demands $d_i$, qualities(CTR) $q_j$ and bids $b_i$}
\KwOut{Equilibrium $(\mathbf{X}$,$\mathbf{p})$}
Sort the bids $b_i$ in decreasing order e.g. $b_1>b_2>\cdots>b_n$\;
Use dynamic programming
\begin{equation}
\label{eq:ce}
g[s,r]=\max\left\{
\begin{array}{l}
g[s-1,r]\\\\
g[s-1,r-d_s]+b_s\sum_{j=r-d_s+1}^{r}q_j\\
\end{array}
\right.
\end{equation}
By tracking dynamic programming find allocation $\mathbf{X}$;\\
Using following LP to settle price $\mathbf{p}$;\\
Let $T_i$ be any consecutive number of $d_i$ slots, for
all $i\in [n]$;
\begin{align*}
  \max \quad & \sum_{i\in [n]}\sum_{j\in X_i}p_j &  \\
  s.t.\quad  & p_j \ge 0 & \forall\ j\in [m] \\
  & p_j = 0 & \forall\ j\notin \cup_{i\in [n]}X_i   \\
  & \sum_{j\in X_i}(v_iq_j-p_j)\ge \sum_{j'\in T_i}(v_iq_{j'}-p_{j'}) & \forall\ i\in [n] \\
  & \sum_{j\in X_i}(v_iq_j-p_j)\ge 0 &\forall i\in[n]
\end{align*}
\If{LP has a feasible solution}{
return $(\mathbf{X}$,$\mathbf{p})$}
\Else{return null\;}
\end{algorithm}

\newpage

\begin{algorithm}[!h]
\SetAlgoLined
\caption{Main Algorithm for CE}
\label{alg:MCE}
\KwIn{Demands $d_i$, qualities(CTR) $q_j$ and bids $b_i$, Accuracy $\epsilon$, biding times K}
\KwOut{R revenue}
$b^1_i=b_i$, $v_i=b_i$ $i=1,2,\cdots,n$.\\
invoke Sub-algorithm for CE on $(d,q,b^1)$,\\
\If{output is not null}{
Suppose the output is $(\mathbf{X}$,$\mathbf{p})$\\
calculate the utility for all $i$. e.g. $u_i=v_i\sum_{j\in X_i}q_j-\sum_{j\in X_i}p_j$}
\For {$r=1,2,\cdots,K$}{
     \For{$i=1,2,\cdots,n$}
     {
       let $M^r_i=\lfloor b^r_i/\epsilon\rfloor$;\\
       \For {$t^r_i=\epsilon,2\epsilon,\cdots,M^r_i*\epsilon$}{
             invoke Sub-algorithm for CE on input $(d,q,(t^r_i,b^r_{-i}))$\\
             \If{the output is not null}{Suppose the output is $(\mathbf{X}$,$\mathbf{p})$\\
             Calculate the current utility $u=v_i\sum_{j\in X_i}q_j-\sum_{j\in X_i}p_j$\\
                 \If{$u>u_i$}{
                  let $u_i=u$ and $b^{r+1}_i=t^r_i$, $b^{r}_i=t^r_i$.\\
                 }
                 \Else{$b_i^{r+1}=b^r_i$;}
             }

        }

    }
  $R^r=\sum_j p_j$
}
\end{algorithm}

\newpage

\subsection{Revenue from generalized GSP}
\begin{algorithm}[!h]
\SetAlgoLined
\caption{Algorithm GSP}
\label{alg:GSP}
\KwIn{Demands $d_i$, qualities(CTR) $q_j$ and bids $b_i$, Accuracy $\epsilon$, biding times K}
\KwOut{R revenue}
$b^1_i=b_i$, $v_i=b_i$ $i=1,2,\cdots,n$.\\
Suppose the allocation of GSP is $\mathbf{X}=sharp(b,q)$;\\
calculate the utility for all $i$. e.g. $u_i=v_i\sum_{j\in X_i}q_j-\sum_{j\in X_i}p_j$\\
\For {$r=1,2,\cdots,K$}{
     \For{$i=1,2,\cdots,n$}
     {
       let $M^r_i=\lfloor b^r_i/\epsilon\rfloor$;\\
       \For {$t^r_i=\epsilon,2\epsilon,\cdots,M^r_i*\epsilon$}{
           Suppose the output of GSP on $(d,q,(t^r_i,b^r_{-i}))$ is $(\mathbf{X},\mathbf{p})$\\
             Calculate the current utility $u=v_i\sum_{j\in X_i}q_j-\sum_{j\in X_i}p_j$ of bidder $i$ \\
                 \If{$u>u_i$}{
                  let $u_i=u$ and $b^{r+1}_i=t^r_i$ $b^{r}_i=t^r_i$ .\\
                 }
                   \Else{$b_i^{r+1}=b^r_i$;}
        }

   }
return $R^r=\sum_j p_j$
}

\end{algorithm}

\newpage

\subsection{Revenue from Envy-free Solution}
Suppose $q_1\ge q_2\ge q_3\ge \cdots\ge q_n$
\begin{algorithm}[!h]
\SetAlgoLined
\caption{Sub-algorithm for EF denoted by EF(d,q,b)}
\label{alg:EF}
\KwIn{Demands $d$, qualities(CTR) $q_j$ and bids $b_i$}
\KwOut{Equilibrium $(\mathbf{X}$,$\mathbf{p})$}
Sort the bids $b_i$ in decreasing order e.g. $b_1>b_2>\cdots>b_n$\;
Use dynamic programming(similar as sharp)(initial values $g[0,0]=0$, $g[1,r]=-\infty$, $r\leq d$)
\begin{equation}
\label{eq:ce}
g[s,r]=\max\left\{
\begin{array}{l}
g[s,r-1]\\\\
g[s-1,r-d]+b_s\sum_{j=r-d+1}^{r}q_j\\
\end{array}
\right.
\end{equation}
By tracking dynamic programming find allocation $\mathbf{X}$;\\
The payment of buyers are $\mathbf{P}$, where $P_i$ is the payment of buyer $i$ ;\\
$P_n=b_n\sum_{j\in X_n}q_j$, and $P_i=b_i(\sum_{j\in X_i}q_j-\sum_{j\in X_{i+1}}q_j)+P_{i+1}$ for $i=1,2,\cdots,n-1$
\end{algorithm}

\newpage
\begin{algorithm}[!h]
\SetAlgoLined
\caption{Main Algorithm for EF}
\label{alg:MEF}
\KwIn{Demands $d$, qualities(CTR) $q_j$ and bids $b_i$, Accuracy $\epsilon$, true value $v_i$, biding times $K$}
\KwOut{R revenue}
$b^1_i=b_i$, $i=1,2,\cdots,n$.\\
invoke Sub-algorithm for EF on $(d,q,b^1)$,\\
\If{output is not null}{
Suppose the output is $(\mathbf{X}$,$\mathbf{P})$\\
calculate the utility for all $i$. e.g. $u_i=v_i\sum_{j\in X_i}q_j-P_i$}
\For {$r=1,2,\cdots,K$}{
     \For{$i=1,2,\cdots,n$}
     {
       let $M^r_i=\lfloor b^r_i/\epsilon\rfloor$;\\
       \For {$t^r_i=\epsilon,2\epsilon,\cdots,M^r_i*\epsilon$}{
             invoke Sub-algorithm for EF on input $(d,q,(t^r_i,b^r_{-i}))$\\
             \If{the output is not null}{Suppose the output is $(\mathbf{X}$,$\mathbf{P})$\\
             Calculate the current utility $u=v_i\sum_{j\in X_i}q_j-P_i$\\
                 \If{$u>u_i$}{
                  let $u_i=u$ and $b^{r+1}_i=t^r_i$, $b^{r}_i=t^r_i$.\\
                 }
                 \Else{$b_i^{r+1}=b^r_i$;}
             }
             \Else{$b_i^{r+1}=b^r_i$;}

        }

    }
  $R^r=\sum_i P_i$
}
\end{algorithm}

\end{document}